\documentclass{llncs}
\usepackage{graphicx,amssymb,amsmath}
\usepackage{lineno}
\usepackage{xspace}
\usepackage{url}
\usepackage{hyperref}
\usepackage{placeins} 

\usepackage[ruled,vlined,linesnumbered,titlenumbered]{algorithm2e}
\usepackage{thm-restate}
\DontPrintSemicolon
\SetKwInput{Input}{Input}
\SetKwInput{Output}{Result}
\SetKw{KwBreak}{break}
\SetKw{Break}{break}
\SetKw{Continue}{continue}
\SetKwRepeat{Do}{do}{while}

\newcommand{\changed}[1]{\textcolor{black}{#1}}
\newcommand{\student}[1]{}
 \usepackage[textsize=scriptsize]{todonotes} 
\renewcommand{\vec}[1]{\overrightarrow{#1}}  

\title{Using Ray-shooting Queries for Sublinear Algo-\\rithms for Dominating Sets in RDV Graphs%
\thanks{Research of TB supported by NSERC, RGPIN-2020-03958.
Research of PG supported by a MITACS Globalink Graduate Fellowship. The results also appeared as part of the second author's Master's thesis \cite{PrashantThesis} \changed{at the University of Waterloo}.} 
} 
\titlerunning{Dominating Sets in RDV graphs}

\author{Therese Biedl\inst{1}
\orcidID{0000-0002-9003-3783}
\and
Prashant Gokhale\inst{2}
}
\institute{David R.~Cheriton School of Computer
Science, University of Waterloo, \\ Waterloo, Ontario, Canada,
\email{biedl@uwaterloo.ca}
\and
University of Wisconsin-Madison, Madison, Wisconsin, USA\\
\email{prashant.gokhale@wisc.edu}
}

\begin{document}

\maketitle
\begin{abstract}
In this paper, we study the dominating set problem in \emph{RDV graphs}, 
\changed{a graph class that lies between interval graphs and chordal graphs
and defined as the}
\textbf{v}ertex-intersection graphs of \textbf{d}ownward paths in a \textbf{r}ooted tree.   
It was shown in a previous paper that adjacency queries in an RDV graph can be reduced to the
question whether a horizontal segment intersects a vertical segment.   This was then used
to find a maximum matching in an \changed{$n$-vertex} RDV graph, 
\changed{using priority search trees, in $O(n\log n)$ time, i.e., without even looking at all edges.}
In this paper, we show that if additionally we also use a ray shooting data
structure, we can also find a minimum dominating set in an RDV graph 
$O(n\log n)$ time (presuming a linear-sized representation of the graph is given).
\changed{The same idea can also be used for a new proof to find a 
minimum dominating set in an interval graph in $O(n)$ time.} \end{abstract}

\section{Introduction}

In a graph \( G \), a subset $S$ of  vertices is called \emph{dominating} if every vertex in the graph is either included in $S$ or adjacent to a vertex in $S$. 
Finding a minimum dominating set is one of the oldest problems in graph theory and graph algorithms, with applications
in routing problems in wireless networks and facility location problems. Finding \changed{a} minimum dominating set is \(\mathsf{NP}\)-hard in general~\cite{GJ79}, as there exists a straightforward reduction from the set cover problem. On the other hand, 
\changed{an easy \(O(\log{n})\)-approximation algorithm}   
is well known \cite{Chvatal1979}. 

In this paper, we study the problem of finding a minimum dominating set in 
graph classes 
\changed{whose}   
structure makes it feasible to find a minimum dominating set very efficiently.   
The minimum dominating set problem remains $\mathsf{NP}$-\changed{h}ard for chordal graphs \cite{BoothJ82}  
but polynomial time algorithms exist for permutation graphs \cite{FarberK85}, cographs \cite{CorneilP84}, and strongly chordal graphs
\cite{Farber84}. 
(Section~\ref{sec:preliminaries} will review those definitions that are needed later again.)
For more information about the complexity of minimum dominating set in various subclasses of perfect graphs, 
the interested reader can look at an overview paper by Corneil and Stewart \cite{CorneilS90}. 

This paper is concerned with finding \emph{sub-linear} time algorithms for minimum dominating set, i.e., we are not even allowed to look at all edges.   The main ingredient here is that the graph is given implicitly, typically via an intersection representation.   For example in an \emph{interval graph} (i.e., the intersection graph of $n$ 1-dimensional intervals), a minimum dominating set can be found in $O(n)$ time if a suitable interval representation is given, even though the graph may have $\Theta(n^2)$ edges.%
\footnote{This result is attributed to \cite{BertossiG88,Bra87,RamalingamR88} in \cite{ChengKZ98}, but neither of these three papers claims an $O(n)$-time algorithm for dominating set in interval graphs.   The closest is \cite{BertossiG88}, which states a run-time of $O(n\log n)$ for the minimum-weight total dominating set problem.   Nevertheless, it is easy to develop an $O(n)$-algorithm; see also Section~\ref{sec:interval}.} 
Other problems can also be solved in sub-linear time in interval graphs, see e.g.~\cite{CPL99} for path cover, Hamiltonian cycle and domatic partition, or \cite{Liang1993} for maximum matching.
Cheng, Kaminski, and Zaks \cite{ChengKZ98} extended the algorithm for dominating set in interval graphs to a graph class slightly bigger than interval graphs, specifically, intersection graphs of downward paths in a \emph{spider} (a tree consisting of multiple paths sharing a common node).

In this paper, we study a graph class called \emph{RDV graphs}, which \changed{strictly} contains the graph class of Cheng, Kaminski and Zaks \changed{(and hence the interval graphs)}. RDV graphs are the intersection graphs of downward paths in an arbitrary rooted tree.    These are known to be strongly chordal graphs, and as such we can find a minimum dominating set in them in linear time \cite{Farber84}.   In fact, this was known even earlier:   Booth and Johnson showed that a straightforward linear-time greedy algorithm \changed{performs optimally if applied to an RDV graph} \cite{BoothJ82}.    

The goal of this paper is to modify the greedy-algorithm by Booth and Johnson to achieve run-time $O(|T|+n\log n)$, where $|T|$ is the size of the host-tree in a given intersection representation of the RDV graph. (We may assume that $|T|\in O(n)$ \changed{\cite{BG-CGT25}}, and will therefore typically state the bound as $O(n\log n)$ time.)   While this is not quite as fast as the algorithm for interval graphs, it is still sublinear since RDV graphs may have quadratically many edges.
The key ingredient is the insight (first published in a companion paper \cite{BG-CGT25}) that
adjacency queries in an RDV graph can be reduced to an intersection problem between horizontal and vertical segments.    In \cite{BG-CGT25},
this insight was combined with range-minimum queries (implemented with a priority search tree) to obtain sublinear algorithms for
maximum matching in RDV graphs.    In this paper, we show that if we \emph{additionally} also use a ray-shooting data structure, 
then more complicated neighbourhood queries can also be solved efficiently without looking at all edges, and use this
to solve the minimum dominating set in sublinear time.    (The technique should work more broadly; we briefly show that we can also
find a minimum dominating set in an interval graph with it.)

The paper is structured as follows.   \changed{In Section~\ref{sec:background},} after reviewing some definitions, we briefly explain results
from the companion paper that we need.   \changed{In Section~\ref{sec:DS}} we give the greedy-algorithm and explain which of its operations cannot easily be sped up with
the known techniques.   Next we explain how to use a 
ray-shooting data structure to implement this operation.  This gives us the faster run-time for the gree\changed{d}y-algorithm for dominating set.
\changed{We discuss an $O(n)$-time algorithm for dominating set in interval graphs in Section~\ref{sec:IS} and
end in Section~\ref{sec:further} with further thoughts.}

\section{Background}
\label{sec:background}
\label{sec:preliminaries}

We assume familiarity with graphs, and refer e.g.~to Diestel's book for basic terminology \cite{Die12}.   For a vertex $z$
in a graph, we write $N[z]$ for the \emph{closed neighbourhood}, i.e., the vertex $z$ and all its neighbours.
Often we will fix a vertex order $z_1,\dots,z_n$.    We then write $N_{\geq}[z_i]$ for those vertices $z_q$ in the 
closed neighbourhood of $z_i$ that satisfy $q\geq i$ (and similarly define $N_<[z_i],N_{\leq}[z_i],N_>[z_i]$).

In this paper we only study a graph $G$
that has a \emph{representation as a vertex-intersection graph of subtrees of a tree} (or \emph{tree-intersection representation} for short). 
This consists of a \emph{host-tree} $T$ and, for each vertex $z$ of $G$, a subtree $T(z)$ of $T$ such that $(z,w)$ is an edge of $G$ if and only 
if $T(z)$ and $T(w)$ share at least one node of $T$.
As convention, we use the term `node' for the vertices of $T$, to distinguish them from the vertices of $G$.
It is well-known \cite{Gavril74}
that a graph has a tree-intersection representation
if and only if it is \emph{chordal}, i.e., it has no induced cycle of length 4 or more.

Numerous subclasses of chordal graphs can be defined by restricting the host-tree and/or the subtrees $T(z)$ in some way.   For example,
we mentioned the (unnamed) class of graphs studied by
Cheng, Kaminski, and Zaks \cite{ChengKZ98} where the host-tree is restricted to be a spider, and the subtrees are restricted to be 
paths that do not contain the center of the spider as interior vertex.   The graph class we study here was named by Monma and Wei \cite{MonmaW86},
though studied even earlier by Gavril \cite{Gavril75}:

\begin{definition} A \emph{rooted directed path graph} (or \emph{RDV graph} \cite{MonmaW86})
is a graph that has an \emph{RDV representation}, i.e., a tree-intersection representation with a rooted host-tree $T$ where for every vertex $z$ the subtree $T(z)$ is a \emph{downward path}, i.e., a path that begins at some node and then always goes downwards in $T$.
\end{definition}

We may assume that $T$ has $O(n)$ nodes \cite{BG-CGT25}.
We assume throughout the paper that an RDV graph $G$ is given implicitly via an RDV representation $(T,t,b)$, i.e., we are given a rooted tree $T$, and for every vertex $z$ of $G$, a \emph{top node} $t(z)$ and a \emph{bottom node} $b(z)$ that is a descendant of $t(z)$, with the convention that $T(z)$ is the path between $t(z)$ and $b(z)$.   
A \emph{bottom-up enumeration}  of $G$ is a vertex order obtained by sorting vertices by decreasing distance of $t(z)$ to the root, breaking ties arbitrarily.
\changed{This can be computed in $O(|T|+n)$ time using bucket-sort.}

\usetikzlibrary{backgrounds}
\begin{figure}[ht]
\centering
\begin{tikzpicture}[
        host/.style = {draw,rectangle,rounded corners,fill=white,inner sep=1pt},
        graph/.style = {draw,circle,fill=white},
        z1/.style = {draw=green!50!black,ultra thick,densely dotted},
        z2/.style = {draw=red,thick,dotted},
        z3/.style = {draw=gray,ultra thick,dash dot dot},
        z4/.style = {draw=orange,thick,dashed},
        z5/.style = {draw=blue!50!red,ultra thick,loosely dotted,cap=round},
        z6/.style = {draw=blue,thick,dash dot},
        z7/.style = {draw=green!50!blue,ultra thick,dash pattern={on 4pt off 1pt on 4pt off 1pt on 1pt off 1pt on 1pt off 1pt}},
        level distance = 1cm,
        sibling distance = 1cm,
]
\begin{scope}[scale=0.45]
\node [draw=none] at (-2,8) (G) {\bf $G$:};
    \node[graph,z3] (z3) at (0, 6) {\( z_3 \)};
    \node[graph,z1] (z1) at (2, 8) {\( z_1 \)};
    \node[graph,z2] (z2) at (0, 8) {\( z_2 \)};
    \node[graph,z4] (z4) at (2, 6) {\( z_4 \)};
    \node[graph,z5] (z5) at (0, 4) {\( z_5 \)};
    \node[graph,z6] (z6) at (2, 4) {\( z_6 \)};
    \node[graph,z7] (z7) at (0, 2) {\( z_7 \)};
    \draw (z3) -- (z5);
    \draw (z5) -- (z6);
    \draw (z6) -- (z7);
    \draw (z6) -- (z4);
    \draw (z4) -- (z1);
    \draw (z5) -- (z4);
    \draw (z2) -- (z4);
    \draw (z2) -- (z1);
\end{scope}
\begin{scope}[scale=1,yshift=5cm,xshift=5cm]
\node [draw=none] at (-1,-0.5) (T) {\bf $T$:};
\draw [->,-latex] (0,-0.5) -- (5.5,-0.5) node [right] {$x$};
\draw [->,-latex] (0,-0.5) -- (0,-4.5) node [left] {$y$};
  \foreach \x/\xlabel in {1/1, 2/2, 3/3, 4/4, 5/5}
    \draw (\x,-0.6) -- (\x, -0.4) node[anchor=south,fill=white] {\small \xlabel};
  \foreach \y/\ylabel in {0/0, -1/1, -2/2, -3/3}
    \draw (0.1,-1+\y) -- (-0.1 ,-1+\y) node[anchor=east, fill=white] {\small \ylabel};

\node [host] (r) at (1,-1) {$\overline{z_6}$, $\overline{z_7}$}
        child {edge from parent[draw=none]}
        child {edge from parent[draw=none]}
        child {edge from parent[draw=none]}
        child {edge from parent[draw=none]}
        child { node [host] (c2) {$\overline{z_4},\overline{z_5},{z_6}$}
                child {edge from parent[draw=none]}
                child {edge from parent[draw=none]}
                child {edge from parent[draw=none]}
                child { node [host] (g3) {$\overline{z_1},\underline{\overline{z_2}},z_4$}
                        child {edge from parent[draw=none]}
                        child { node [host] (b3) {$\underline{z_1}$}}
                        child { node [host] (l4) {$\underline{z_4}$}}
                }
                child {edge from parent[draw=none]}
                child { node [host] (g2) {$\overline{z_3},z_5$}
                        child { node [host] (l2) {$\underline{z_3},\underline{z_5}$}}
                }
                child { node [host] (l5) {$\underline{z_6}$} }
        }
        child {edge from parent[draw=none]}
        child {edge from parent[draw=none]}
        child {edge from parent[draw=none]}
        child { node [host] (c1) {$z_7$}
                        child { node [host] (l1) {$\underline{z_7}$}}
        }
;
\end{scope}

\begin{scope}[on background layer]
\draw [z3] ([xshift=-5pt]g2.center)-- ([xshift=-5pt]l2.center); 
\draw [z1] ([xshift=-5pt]g3.center)-- ([xshift=-5pt]b3.center); 
\draw [z2] ([xshift=2pt]g3.north east) -- ([xshift=2pt]g3.south east); 
\draw [z4] ([xshift=5pt]c2.center) -- ([xshift=5pt]g3.center)-- ([xshift=5pt]l4.center); 
\draw [z5] ([xshift=5pt]c2.center) -- ([xshift=5pt]g2.center) -- ([xshift=5pt]l2.center); 
\draw [z6] ([xshift=5pt]r.center) -- ([xshift=5pt]c2.center) -- ([xshift=5pt]l5.center); 
\draw [z7] ([xshift=5pt]r.center) -- ([xshift=5pt]c1.center) -- ([xshift=5pt]l1.center); 
\end{scope}
\end{tikzpicture}
\caption{
An RDV graph with an RDV representation. 
Each node $\nu$ lists those vertices $z$ with $\nu\in T(z)$; we write $\underline{z}$ if $\nu=b(z)$ and $\overline{z}$ if $\nu=t(z)$.   
We also show the paths as poly-lines, with colors/dash-pattern matching the vertices.
Nodes are drawn at their coordinates, and indices correspond to a bottom-up enumeration order. 
}
\label{fig:RDV_example}
\end{figure}
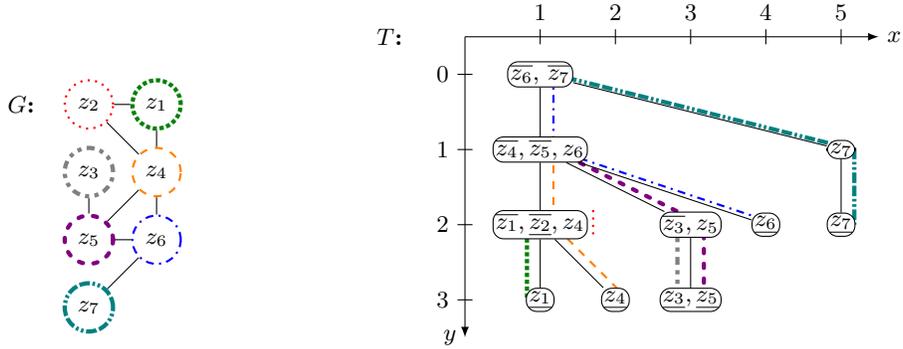

\subsection{From RDV representation to segments}

We briefly review some crucial insights taken from our companion paper \cite{BG-CGT25}.
Let $G$ be a graph with an RDV representation with host tree $T$.
Assign points to the nodes of the host-tree $T$ as follows:
For each node $\nu$, define $y$-coordinate $y(\nu)$ to be the distance of node $\nu$ from the root of $T$.
To define the $x$-coordinate, first
fix an arbitrary order of children at each node, then enumerate the leaves of $T$ as $L_1,\dots,L_\ell$ from left to right,
and then for each node $\nu$ set $x(\nu)$ to be the index of the leftmost leaf that is a descendant of $\nu$.   
Figure~\ref{fig:RDV_example} 
shows each node $\nu$ drawn at point $(x(\nu),y(\nu))$ (where $y$-coordinates increase top-to-bottom).    
We can compute these coordinates in $O(|T|)$ time.

\begin{definition}
Let $G$ be a graph with an RDV representation. For each vertex~$z$, define the following (see also Figure~\ref{fig:RDV_horsegment}): 

\begin{itemize}
\item The \emph{horizontal segment} $\mathbf{h}(z)$ of $z$ is the segment that
extends rightward from the point of $t(z)$ until 
\changed{the $x$-coordinate of} 
the rightmost descendant of $t(z)$. 
\item The \emph{vertical segment} 
$\mathbf{v}(z)$ of $z$ is the vertical segment that extends upward 
from the point of $b(z)$ until 
\changed{the $y$-coordinate of} 
of $t(z)$ (i.e., \changed{it} hits $\mathbf{h}(z)$). 
\end{itemize}
\end{definition}

Crucial to developing efficient algorithms is that edges in the graph can be characterized via intersections of these segments.

\begin{theorem}\cite{BG-CGT25}
\label{thm:RDVgeometric}
Let $G$ be a graph with an RDV representation and let $z_1,\dots,z_n$ be a bottom-up enumeration of vertices.     Then for any $i<j$,
edge $(z_i,z_j)$ exists if and only if the vertical segment $\mathbf{v}(z_j)$ intersects the horizontal segment $\mathbf{h}(z_i)$.
\end{theorem}

\subsection{Pre-processing the RDV representation}

For our data structures to come, it will be helpful if each set of segments defined by vertices are disjoint.   This could be achieved with perturbing the coordinates,  but we find it easier to modify instead the RDV representation to have some properties that may be of independent interest.  The following is easy to show (see Figure~\ref{fig:RDV_horsegment} for an example).

\begin{restatable}{lemma}{NiceRDV}
\label{lem:nice}
Let $G$ be a graph with an RDV representation $(T,t,b)$. Then, in $O(|T|+n)$ time, we can conver\changed{t} $(T,t,b)$ into an RDV representation $(T',t',b')$ such that
$t'(z)\neq t'(w)$ and $b'(z)\neq b'(w)$ for all \changed{pairs of vertices} $z\neq w$.   Furthermore, all bottom nodes are leaves and no top node is a leaf.   
\end{restatable}

\begin{proof}
We assume that each node $\nu$ of $T$ knows all vertices for which it is the top node, and all vertices for which it is the bottom node.   (We can compute this in $O(|T|+n)$ time.)   Now perform the following at each node $\nu$ of $T$.   First, if there are $d$ vertices $z_1,\dots,z_d$ for which $\nu$ is the bottom node, then add $d$ new leaves as children of $\nu$, and make these the new bottom nodes for $z_1,\dots,z_d$.   (With this, no top node is a leaf anymore, since the corresponding bottom node became a strict descendant.)    Second, if there are $k\geq 2$ vertices $w_1,\dots,w_k$ for which $\nu$ is the top node, then let $\rho$ be the parent of $\nu$ (temporarily insert a new node as parent $\rho$ if $\nu$ was the root).   Subdivide the link from $\nu$ to $\rho$ with $k-1$ new nodes, and make these the new top nodes of $w_2,\dots,w_k$.   One easily verifies that the newly defined paths intersect in a node if and only if the old paths did.   We added at most two new nodes for each vertex of $G$, and so $|T'|\leq |T|+O(n)$.
\end{proof}

\usetikzlibrary{backgrounds}
\begin{figure}[ht]
\centering
\begin{tikzpicture}[
        host/.style = {draw,rectangle,rounded corners,fill=white,inner sep=1pt},
        graph/.style = {draw,circle,fill=white},
        v1/.style = {draw=brown,ultra thick},
        z3/.style = {draw=gray,ultra thick,dash dot dot},
        z1/.style = {draw=green!50!black,ultra thick,densely dotted},
        z2/.style = {draw=red,thick,dotted},
        z4/.style = {draw=orange,thick,dashed},
        z5/.style = {draw=blue!50!red,ultra thick,loosely dotted,cap=round},
        z6/.style = {draw=blue,thick,dash dot},
        z7/.style = {draw=green!50!blue,ultra thick,dash pattern={on 4pt off 1pt on 4pt off 1pt on 1pt off 1pt on 1pt off 1pt}},
        v9/.style = {draw=red!20!blue!30!white,line width=3pt,dash pattern={on 1pt off 1pt on 1pt off 1pt}},
        level distance = 1cm,
        sibling distance = 1cm,
]

\begin{scope}[scale=1,yshift=0cm,xshift=-5cm]
\useasboundingbox (0,0) rectangle (8,-7.5);
\draw [->,-latex] (0,-0.5) -- (7.5,-0.5) node [right] {$x$};
\draw [->,-latex] (0,-0.5) -- (0,-7.5) node [left] {$y$};
  \foreach \x/\xlabel in {1/1, 2/2, 3/3, 4/4, 5/5, 6/6, 7/7}
    \draw (\x,-0.6) -- (\x, -0.4) node[anchor=south,fill=white] {\small \xlabel};
  \foreach \y/\ylabel in {0/0, -1/1, -2/2, -3/3, -4/4, -5/5, -6/6}
    \draw (0.1,-1+\y) -- (-0.1 ,-1+\y) node[anchor=east, fill=white] {\small \ylabel};

\node [host] (t7) at (1,-1) {${\overline{z_7}}$}
    child {edge from parent[draw=none]}
    child {edge from parent[draw=none]}
    child {edge from parent[draw=none]}
    child {edge from parent[draw=none]}
    child {edge from parent[draw=none]}
    child {edge from parent[draw=none]}
    child { node [host] (t6) {$\overline{z_6}$, ${z_7}$}
	child {edge from parent[draw=none]}
	child {edge from parent[draw=none]}
	child {edge from parent[draw=none]}
	child {edge from parent[draw=none]}
	child {edge from parent[draw=none]}
	child {edge from parent[draw=none]}
        child { node [host] (t5) {$\overline{z_5},{z_6}$}
            child {edge from parent[draw=none]}
            child {edge from parent[draw=none]}
            child {edge from parent[draw=none]}
            child {edge from parent[draw=none]}
            child {edge from parent[draw=none]}
            child { node [host] (t4) {$\overline{z_4},{z_5},{z_6}$}
                child {edge from parent[draw=none]}
                child {edge from parent[draw=none]}
                child {edge from parent[draw=none]}
                child {edge from parent[draw=none]}
                child {edge from parent[draw=none]}
                child { node [host] (t2) {${\overline{z_2}},z_4$}
                    child {edge from parent[draw=none]}
                    child {edge from parent[draw=none]}
                    child { node [host] (t1) {$\overline{z_1},{{z_2}},z_4$}
                        child {edge from parent[draw=none]}
                        child {edge from parent[draw=none]}
                        child { node [host] (b3) {$\underline{z_1}$}}
                        child { node [host] (b4) {$\underline{z_2}$}}
                        child { node [host] (b5) {$\underline{z_4}$}}
                    }
                    child {edge from parent[draw=none]}
                    child { node [draw=none] (t1next) {} edge from parent[draw=none]}
                }
                child {edge from parent[draw=none]}
                child { node [draw=none] (t2next) {} edge from parent[draw=none]}
                child { node [host] (t3) {$\overline{z_3},z_5$}
                        child { node [host] (l2) {${z_3},{z_5}$}
                		child {edge from parent[draw=none]}
                        	child { node [host] (b2) {$\underline{{z_3}}$}}
                        	child { node [host] (b6) {$\underline{z_5}$}}
			}
                }
                child { node [draw=none] (t3next) {} edge from parent[draw=none]}
                child { node [host] (g7) {${z_6}$} 
                	child { node [host] (b7) {$\underline{z_6}$} }
		}
            }
            child {edge from parent[draw=none]}
            child {edge from parent[draw=none]}
            child {edge from parent[draw=none]}
            child {edge from parent[draw=none]}
            child { node [draw=none] (t4next) {} edge from parent[draw=none]}
        }
	child {edge from parent[draw=none]}
	child {edge from parent[draw=none]}
	child {edge from parent[draw=none]}
	child {edge from parent[draw=none]}
        child { node [draw=none] (t5next) {} edge from parent[draw=none]}
        child { node [host] (c1) {$z_7$}
                        child { node [host] (b8) {$\underline{z_7}$}}
	}
    }
    child {edge from parent[draw=none]}
    child {edge from parent[draw=none]}
    child {edge from parent[draw=none]}
    child {edge from parent[draw=none]}
    child {edge from parent[draw=none]}
    child { node [draw=none] (t6next) {} edge from parent[draw=none]}
;
\node [draw=none] (t7next) at (7,-1) {};
\end{scope}

\begin{scope}[on background layer] 
\draw [z1] ([yshift=-1mm]t1.south west) -- ([xshift=1mm,yshift=-1.5mm]t1next.south east); \draw [z1] ([xshift=1mm,yshift=-1.5mm]t1.south) -- ([xshift=1mm]b3.north); 
\draw [z2] ([yshift=-1mm]t2.south west) -- ([xshift=1mm,yshift=-1.5mm]t2next.south east); \draw [z2] ([xshift=11mm,yshift=-1.5mm]t2.south) -- ([xshift=1mm]b4.north); 
\draw [z3] ([yshift=-1mm]t3.south west) -- ([xshift=1mm,yshift=-1.5mm]t3next.south east); \draw [z3] ([xshift=1mm,yshift=-1.5mm]t3.south) -- ([xshift=1mm]b2.north); 
\draw [z4] ([yshift=-1mm]t4.south west) -- ([xshift=1mm,yshift=-1.5mm]t4next.south east); \draw [z4] ([xshift=21mm,yshift=-1.5mm]t4.south) -- ([xshift=1mm]b5.north); 
\draw [z5] ([yshift=-1mm]t5.south west) -- ([xshift=1mm,yshift=-1.5mm]t5next.south east); \draw [z5] ([xshift=41mm,yshift=-1.5mm]t5.south) -- ([xshift=1mm]b6.north); 
\draw [z6] ([yshift=-1mm]t6.south west) -- ([xshift=1mm,yshift=-1.5mm]t6next.south east); \draw [z6] ([xshift=51mm,yshift=-1.5mm]t6.south) -- ([xshift=1mm]b7.north); 
\draw [z7] ([yshift=-1mm]t7.south west) -- ([xshift=1mm,yshift=-1.5mm]t7next.south east); \draw [z7] ([xshift=1mm,yshift=-1.5mm]t7next.south) -- ([xshift=1mm]b8.north); 
\end{scope}
\end{tikzpicture}
\caption{The modified RDV representation, and the segments corresponding to vertices (drawn slightly offset for legibility).}
\label{fig:RDV_horsegment}
\end{figure}
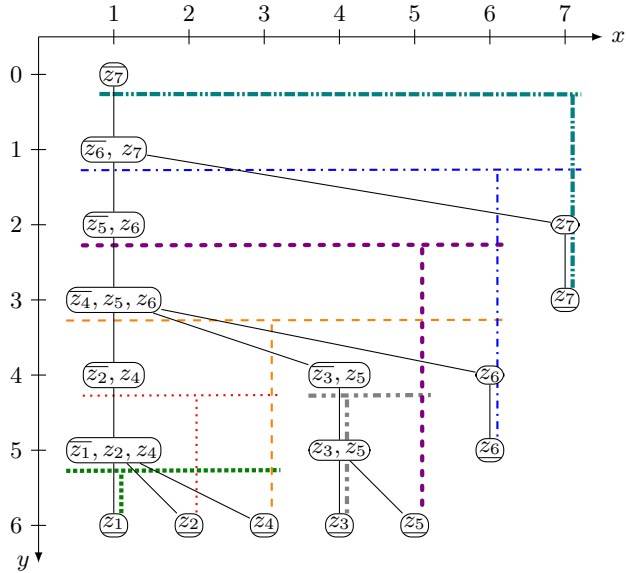

From now on, we assume that the given RDV representation satisfies the conditions of Lemma~\ref{lem:nice}.   With this, $\mathbf{v}(z)$ is disjoint from $\mathbf{v}(w)$ for any $z\neq w$, since $b(z)\neq b(w)$ are both leaves, hence obtain distinct $x$-coordinates, and this determines the $x$-coordinates of the vertical segments.   Likewise, $\mathbf{h}(z)$ is disjoint from $\mathbf{h}(w)$\changed{: If} $t(z),t(w)$ have distinct $y$-coordinates, then so do the two horizontal segments.   If $t(z)$ and $t(w)$ have the same $y$-coordinate, then they are on the same level of the host tree, and by $t(z)\neq t(w)$ they have disjoint sets of descendants.   Since the $x$-coordinates of the descendants determine the $x$-range of the horizontal segments, this makes $\mathbf{h}(z)$ and $\mathbf{h}(w)$ disjoint.   

\section{Dominating Sets in RDV Graphs}
\label{sec:DS}

To find a minimum dominating set in an RDV graph efficiently, we use
a well-known greedy-algorithm \cite{BoothJ82}.
Assume that we are given a vertex order $z_1,\dots,z_n$.
Initialize the dominating set $D$ as empty and mark all vertices as ``undominated''.
Then, for $i=1,\dots,n$, proceed as follows:
\begin{quotation}
If $z_i$ is ``undominated'', then select the vertex $z_q\in N[z_i]$ that maximizes
the index $q$ in the vertex order, add $z_q$ to $D$, and mark all neighbours of $z_q$
as ``dominated''.
\end{quotation}
Clearly this gives a dominating set, and as Booth and Johnson showed \cite{BoothJ82},
it gives a minimum dominating set if run on an RDV graph \changed{using} a bottom-up elimination
order \changed{as the vertex order}.   Also, the algorithm clearly can be implemented in linear time since we
spend $O(\deg(z))$ time whenever a vertex $z$  becomes dominated for the first time, 
or is added to $D$, and everything else takes $O(n)$ time in total.

To reduce the run-time to sublinear \changed{in the number of edges} (and specifically, $O(n\log n)$ after computing
the coordinates of host-tree $T$), we separate out the crucial operations that we
need.

\begin{itemize}
    \item [A.] \changed{For} vertex $z_i$, find the vertex $z_q \in N[z_i]$ that has the maximal index $q$. 
    \item [B.] Given \changed{a} vertex $z_q$, ensure that all vertices in $N[z_q]$ are marked ``dominated''.
\end{itemize}

Note that each operation is performed at most once per iteration, hence $n$ times in total,
so it would be sufficient to our run-time to show how to implement them in $O(\log n)$ time.
(We will achieve this bound only in an amortized sense: We use $O(n \log n)$ time over all iterations
together.)

Operation~A is very similar to an operation that we needed in our companion paper \cite{BG-CGT25}
(the only difference is that we use the closed neighbourhood, rather than the open one, and we 
maximize rather than minimize).   
As shown there, it can therefore be implemented in $O(\log n)$
time per operation using a priority search tree; we review this briefly below.

\changed{However, this implementation works efficiently only for the ``current vertex'' $v_i$, and therefore
cannot be used for Opration~B, where the input-vertex $z_q$ is arbitrary.}
(We know \changed{that} $q\geq i$ since $z_i\in N[z_i]$ \changed{and we maximize the index}, but $q>i$ is possible.)
As such, we need to develop a new and entirely different approach to support Operation~B.

\subsection{Implementing Operation $A$ in $O(\log{n})$ Time per Iteration} \label{oprA}

Define for each vertex $z$ of $G$ a \emph{tuple} consisting of a \emph{value} and a \emph{weight}.   We set the value to be $x(b(z))$, i.e., the $x$-coordinate of $\mathbf{v}(z)$.   We set the weight to be the index that $z$ has in the vertex order, i.e., it is $q$ if $z=z_q$.   To support Operation~A, we store some of these tuples in a \emph{priority search tree} $\mathcal{P}$, a data structure for 2-dimensional points that was introduced by McCreight \cite{pst}.   
This supports (among others) the operation \emph{range-maximum-query}, which receives as input a range $[x,x']$ of values, and it returns, among all tuples that are stored in $\mathcal{P}$ and whose value falls into the range, the one that maximizes the weight (breaking ties arbitrarily, and returning ``null'' if there is no such tuple).   It can do so in $O(\log |\mathcal{P}|)$ time \cite{pst}.    (Our priority search tree always stores at most $n$ tuples, so this is $O(\log n)$ time.)   A priority search tree also supports the addition or deletion of a tuple in $O(\log n)$ time.   With this, we have all tools that are needed for Operation~A.

\begin{lemma} 
\label{opAlemma}
We can implement Operation~A in amortized $O(\log n)$ time.
\end{lemma}
\begin{proof}
We maintain a priority search tree $\mathcal{P}$, initially empty.  We will ensure that throughout the following holds for $\mathcal{P}$:

\smallskip
\textbf{Invariant:\label{inv:B}} In iteration $i$, the priority search tree $\mathcal{P}$ stores exactly those tuples of vertices $z_q$ for which $q\geq i$ and $y(b(z_q)) \geq y(t(z_i)) \geq y(t(z_q))$, i.e., the $y$-range of $\mathbf{v}(z_q)$ includes $y(t(z_i))$.
\smallskip

So at the beginning of iteration $i$, we add to $\mathcal{P}$ all those tuples of vertices $z_p$ where we newly have $y(b(z_p))\geq y(t(z_i))$,
which implies $i=1$ or $y(b(z_p))<y(t(z_{i-1}))$.
(We can find these vertices efficiently by pre-sorting a list $L$ of all vertices by the $y$-coordinate of the bottom-node, and 
transferring a vertex from $L$ to $\mathcal{P}$ once this coordinate is big enough.)
A vertex $z_p$ that was newly added to $\mathcal{P}$ satisfies $p\geq i$, since either $i=1$ or
\changed{$y(t(z_p))\leq y(b(z_p))<y(t(z_{i-1}))$. H}ence $p>i{-}1$ since we use a bottom-up elimination order.
At the end of iteration $i$, we remove $z_i$'s tuple from $\mathcal{P}$, hence all tuples left in $\mathcal{P}$ have index bigger than $i$ and the invariant holds for the next iteration.
Since $y(t(z_i))$ can only decrease as $i$ increases, every vertex is added to $\mathcal{P}$ exactly once. It also is deleted exactly once, and 
the total time to maintain $\mathcal{P}$ over all iterations is $O(n\log n)$.

Recall that Operation~A means finding (in iteration $i$) the vertex $z_q\in N[z_i]$ that 
maximizes index $q$.   Since $z_i\in N[z_i]$, we have $q\geq i$, so we only need to search among vertices in $N_{\geq}[z_i]$.   We claim that doing so
is the same as to perform a range-maximum query in $\mathcal{P}$ with the $x$-range of $\mathbf{h}(z_i)$.     
Assume first that the range-maximum query returns the tuple of $z_q$.    Then the $x$-range of $\mathbf{h}(z_i)$ contains
$x(b(z_q))$, the $x$-coordinate of $\mathbf{v}(z_q)$.   Also $q\geq i$ and the $y$-coordinate of $\mathbf{h}(z_i)$ falls into the $y$-range of $\mathbf{v}(z_q)$ by the invariant.   Hence $\mathbf{v}(z_q)$ intersects $\mathbf{h}(z_i)$, and $z_q\in N_{\geq}[z_i]$.   Similarly one argues that any vertex $z_q\in N_{\geq}[z_i]$ gives an element of $\mathcal{P}$ whose value falls into the range.  Therefore range-maximum query considers exactly the (tuples of) vertices of $N_{\geq}[z_i]$ and returns the one with the largest index, as required.
So performing Operation~A means one range-minimum query which takes $O(\log n)$ time.
\end{proof}

\changed{I}n our example of Figure~\ref{fig:RDV_horsegment}, in the first iteration we have $y(t(z_1))=5$ and $\mathcal{P}$ stores the
tuples of $z_1,\dots,z_6$ since their vertical segments span across $y$-coordinate 5.   We search for $x$-range $[2,4]$, and the vertical
segments of $z_1,z_2,z_4$ fall into this range.   Of these vertices, $z_4$ has the largest weight and we add it to the dominating set.

\subsection{Implementing Operation $B$ in $O(\log{n})$ Time per Iteration} \label{oprB}

In Operation $B$ (in iteration $i$), we are given a vertex $z_q$ with $q\geq i$, and we must ensure that
its closed neighbourhood is marked dominated.   We break this up into two steps, considering $N_{\leq}[z_q]$
and $N_{>}[z_q]$ separately.
As discussed earlier, we cannot use the priority search tree
idea for this, since the segments stored in the priority search tree depend on the $y$-coordinate of $t(z_i)$, and so it can only
be used for a search within $N[z_i]$.   We instead store two different data structures (one for horizontal and one for vertical
segments) that support \emph{ray shooting queries}.

We need some background material first.    Let $S$ be a set of disjoint horizontal segments.  A \emph{ray shooting query}
receives as input a vertical ray $\vec{q}$ specified by its endpoint $q$ and the direction in which it emanates.   The desired output is
the first segment of $S$ that is \emph{hit} by the ray, i.e., it contains a point $p$ of $\vec{q}$, and the open segment $\overline{qp}$
intersects no segment of $S$.   (We return ``null'' if $\vec{q}$ hits no segments of $S$.)
Many data structures have been developed for answering ray shooting queries efficiently, and that 
permit the insertion or deletion of segments.
See for example \cite{KaplanMT03,1990MehlhornNaher} for older results with slower processing time. 
The current best run-time bounds are by Giyora and Kaplan \cite{GiyoraK09}; their data structure uses $O(|S|)$ space and supports 
updates and queries in $O(\log |S|)$ time.   (In our setting, we will always have $|S|\leq n$.)
It was shown later that these run-time bounds can be achieved even when segments and rays are not orthogonal to each other
\cite{2021Nekrich}, though we do not need this for our application.

\begin{lemma} 
\label{opBlemma}
We can implement Operation~B in amortized $O(\log n)$ time.
\end{lemma}
\begin{proof}
We explain here first how to mark all undominated vertices in $N_\leq[z_q]$ and then briefly
sketch how to handle $N_>[z_q]$ symmetrically.
We maintain throughout
the algorithm a set $S_H$ that stores the horizontal segments of all those vertices that are \emph{not} dominated.   (It is
very important to remove segments of dominated vertices from this set for the amortized analysis to work out.)   Initially
$S_H$ simply stores $\mathbf{h}(z)$ for all vertices $z$.   (Recall that these segments are disjoint.)
We remove $\mathbf{h}(z)$ from $S_H$ when we
first mark $z$ as ``dominated''.   Clearly this takes $O(n\log n)$ time in total, and hence $O(\log n)$ \changed{amortized} time per iteration.

We define for every vertex $z$ a ray $\vec{z}$ that starts at the top end of $\mathbf{v}(z)$ and shoots along it
(i.e., towards $\infty$ in $y$-coordinates).    Assume that in iteration $i$ we have fixed vertex $z_q$, and
perform a ray-shooting query with ray $\vec{z_q}$ within the segment-set $S_H$.      
If the query returns $\mathbf{h}(z_p)$, then we explicitly \changed{check} whether $z_p\in N[z_q]$; since we can argue \changed{that} $p\leq q$ (see below)
this takes constant time using Theorem~\ref{thm:RDVgeometric}.   If we indeed have $z_p\in N_{\leq}(z_q)$ then 
we mark $z_p$ as dominated (it was ``undominated'' before since its segment was in $S_H$).  We remove its segment 
from $S_H$ (and also from the ``other'' data structure $S_V$
that we define below).   Then repeat the ray-query.

To see that we must have $p\leq q$, recall that we shoot the ray downwards (towards larger $y$-coordinates),
so $y(t(z_p))\geq y(t(z_q))$.   So we could have $p>q$ only if $y(t(z_p))= y(t(z_q))$, but since top nodes are distinct
this then means that the $x$-range of $\mathbf{h}(z_p)$ is disjoint from the one of $\mathbf{h}(z_q)$.   The latter range
contains the $x$-coordinate of the ray, so then $\vec{z_q}$ cannot possibly have hit $\mathbf{h}(z_p)$.  

Assume now that the query returns $\mathbf{h}(z_p)$ for some vertex $z_p$ where $p\leq q$ but $z_p\not\in N[z_q]$.
We claim that then we are done searching for neighbours.   (For the search in $N_{\leq}[z_q]$ this case actually
cannot happen, but for the search in $N_>[z_q]$ this may happen.)   Namely, the intersection point between \changed{the} ray
and $\mathbf{h}(z_p)$ then is \emph{outside} the $y$-range of $\mathbf{v}(z_q)$ (otherwise there would be an
edge by Theorem~\ref{thm:RDVgeometric}).    Further queries can only find segments where the $y$-coordinate is
even bigger, hence also outside the $y$-range of $\mathbf{v}(z_q)$, and so will not reveal any more elements of $N_{\leq}[z_q]$.

So if we perform ray-shooting queries until we either receive ``null'' or find a vertex not in the neighbourhood of $z_q$,
then 
\changed{all found vertices are in}
$N_\leq[z_q]$.    Vice versa, any undominated vertex in $N_\leq[z_q]$ will be found, since its
horizontal segment intersects $\mathbf{v}(z_q)$ and so it will be hit by the ray.   
\changed{O}n the example
of Figure~\ref{fig:RDV_horsegment}, in the first iteration (where we had $z_q=z_4$), we shoot a vertical downward
ray from point $(3,3)$, which hits the horizontal segments of $z_4,z_2,z_1$ (in this order) and then returns ``null''.

If we found $k$ vertices in $N_\leq[z_q]$ (for some $k\geq 0$), then we had to perform $k{+}1$ ray-shooting queries, which
took $O((k+1)\log n)$ time.   But we also marked $k$ vertices as dominated.  We can hence attribute one of
these log-factors to the iteration $i$ that triggered Operation~B, and all others to the vertices that are for
the first time marked as ``dominated''.   This happens once per vertex, and hence is $O(\log n)$ amortized overall.

\medskip
In a symmetric way we can ensure that all vertices in $N_>[z_q]$ are marked as dominated:   Store
a data structure $S_V$ that stores $\mathbf{v}(z)$ for every undominated vertex $z$.
(Recall that these segments are disjoint.)
To find $N_>[z_q]$, shoot a ray from the left endpoint of $\mathbf{h}(z_q)$ rightward, and find the first
segment in $S_V$ that it intersects.   If this intersects, say, $\mathbf{v}(z_p)$, then first test whether 
the $x$-coordinate of the intersection point lies within $\mathbf{h}(z_q)$.   If not, then by Theorem~\ref{thm:RDVgeometric}
(and because we return the first segment that was hit) we have found all undominated vertices in $N_>[z_q]$.   If
the $x$-coordinate is within range, then $z_p$ is an undominated neighbour of $z_q$ (and necessarily in $N_>[z_q]$
since we found all undominated vertices in $N_\leq[z_q]$ earlier already and removed their segments from $S_V$).
We mark it as dominated and repeat the ray-shooting until it returns ``null''.   
In our example, we would shoot a rightward ray from point $(1,3)$.
\changed{The first segment hit by the ray belongs to $z_5$ (remember that $z_4$ was marked dominated
earlier, and its segment removed from $S_V$ in consequence).   We mark $z_5$ as dominated, and likewise
$z_6$ (whose segment is hit next).   Next the ray hits the vertical segment of $z_7$, but here}
the $x$-coordinate of the intersection-point is  outside the range of $\mathbf{h}(z_4)$;
we stop the search and have found all of $N_>[z_4]$.
\end{proof}

\begin{algorithm}[ht]
\caption{Dominating set in RDV graphs}
\label{alg:reformulated_domset}
\Input{A graph $G$ with an RDV representation $T$}
Compute the coordinates of nodes of $T$. \\
Bucket-sort the vertices by decreasing $y(t(\cdot))$ to obtain vertex order $z_1,\dots,z_n$. \\
Bucket-sort the vertices by decreasing of $y(b(\cdot))$ to obtain a list $L$.\\
Initialize dominating set $D = \emptyset$ and empty priority search tree $\mathcal{P}$ \\
Initialize orthogonal ray shooting structures $S_H$ and $S_V$\\

\For{$i = 1,2,\cdots, n$}{
$S_H.\text{insert}(\mathbf{h}(z_i))$,
$S_V.\text{insert}(\mathbf{v}(z_i))$,
mark $z_i$ as ``undominated''.
}

\For{$i = 1,2,\cdots, n$}{
	\While(\tcp*[f]{Add to $\mathcal{P}$}){the first vertex $z_p$ in $L$ satisfies $y(b(y_p))\geq y(t(z_i))$}{
		delete $z_p$ from $L$ and add its tuple $[x(b(z_q)),q]$ to $\mathcal{P}$
	}

	\If{$z_i$ is ``undominated''}{

       	 $[\cdot,q] \leftarrow \mathcal{P}.\text{rangemax}(\text{$x$-range of $\mathbf{h}(z_i))$}$ \tcp*{Operation A}
       	 $D \leftarrow D \cup \{ z_q \}$ \;

       	 $\vec{z_q}\leftarrow$ downward ray from top end of $\mathbf{v}(z_q)$ \tcp*{Operation B, $N_{\leq}$ }
       	 $s = S_H.\text{orthrayshoot}(\vec{z_q})$\;
       	 \While{$s$ is not null (say $s=\mathbf{h}(z_p)$)}
       	 {
		\uIf{$y(t(z_p))$ lies in the $y$-range of $\mathbf{v}(z_q)$} {
	       	     Mark $z_p$ as ``dominated'',
	       	     $S_H.\text{delete}\left(\mathbf{h}(z_p)\right)$,
       		     $S_V.\text{delete}\left(\mathbf{v}(z_p)\right)$\;
       		     $s \leftarrow S_H.\text{orthrayshoot}(\vec{z_q})$
		}
		\lElse { $s\gets$ null} 
       	 }

       	 $\vec{z_q}\leftarrow$ rightward ray from left end of $\mathbf{h}(z_q)$ \tcp*{Operation B, $N_>$}
        $s = S_V.\text{orthrayshoot}(\vec{z_q})$\;
       	 \While{$s$ is not null (say $s=\mathbf{v}(z_p)$)}
        {
	   \uIf{$x(b(z_p))$ lies in the $x$-range of $\mathbf{h}(z_q)$}{
	            Mark $z_p$ as ``dominated'',
       	     $S_H.\text{delete}\left(\mathbf{h}(z_p)\right)$,
       	     $S_V.\text{delete}\left(\mathbf{v}(z_p)\right)$\;
       	     $s \leftarrow S_V.\text{orthrayshoot}(\vec{z_q})$
		}
		\lElse { $s\gets$ null}
        }

	}
    delete tuple of $z_i$ from $\mathcal{P}$
}
\KwRet{$D$}
\end{algorithm}

\bigskip
This ends the description of how to implement the operations efficiently, and gives us our main result
(we also summarize the algorithm with the pseudocode in Algorithm~\ref{alg:reformulated_domset}).

\begin{theorem} \label{domsetruntime}
\label{thm:main}
Given an $n$-vertex graph $G$ with an $O(n)$-sized RDV representation $T$, a minimum dominating set of $G$ can be found in $O(n\log n)$ time.
\end{theorem}

The $O(\log n)$ amortized run-time is achieved with 
the ray shooting data structure of Giyora and Kaplan \cite{GiyoraK09}. This uses dynamic fractional cascading and a generalized version of van Emde Boas trees, and hence is quite complex. One could instead use an older ray-shooting data structure of Overmars \cite{Overmars85} which is considerably simpler. However, this uses run-time $O(\log^2{n})$ for updates and queries, so we would incur an extra $\log{n}$ factor in the runtime bounds.

%

\FloatBarrier

\section{$O(n)$-time Algorithm for Interval Graphs}
\label{sec:interval}
\label{sec:IS}

As mentioned in the introduction, an algorithm to find a minimum dominating set
in an interval graph
with run-time $O(n)$ was known to Cheng et al.~\cite{ChengKZ98}, but we have not
been able to find a description of it in the literature.   Since it is very easy
to state such an algorithm using the ideas of the previous section, we therefore
give an independent proof of this result here.

Every interval graph is an RDV-graph, because we can view its representation by
intersecting intervals as an RDV-representation where the host tree $T$
is a path (rooted at one end), and the interval of vertex $v$ becomes a
path $T(z)$ connecting its top end $t(z)$ to its bottom end $b(z)$.     With
operations similar to Lemma~\ref{lem:nice}, we
can achieve that $t(z)\neq t(w)$ 
for any two vertices $z\neq w$
while keeping \changed{the} host tree \changed{a path} of size $O(|T|+n)$.
Since $T$ is a path, \emph{all} nodes of $T$ have the same
$x$-coordinate, and only the $y$-coordinates are relevant for intersections.
For ease of description, we will identify nodes of $T$ with their $y$-coordinates
(i.e., distance from the root), and so for example write $[t(z),b(z)]$ for the interval 
of vertex $v$, and statements such as $b(w)\leq t(z)$.

To find a minimum dominating set in an interval graph, we use the
same greedy-algorithm, but implement the two required operations differently.

\medskip\noindent{\bf Operation A:}   We maintain a set $\mathcal{P}$ with the
same invariant, but can use a simpler data structure since all $x$-coordinates are the same:
We store $\mathcal{P}$ as an unsorted list of vertices.   Clearly we can add to this and delete from
it in $O(1)$ time, presuming every vertex keeps track of where
it is stored in \changed{$\mathcal{P}$}.

In iteration~$i$, list $\mathcal{P}$ stores exactly those vertices $z_p$ with
$t(z_p)\leq t(z_i)\leq b(z_p)$.   All these vertices intersect the interval of $z_i$
(at $t(z_i)$), and $p\geq i$, so $\mathcal{P}$ stores exactly $N_{\geq}(z_i)$.
To find the vertex $z_q\in N_{\geq}(z_i)$ that maximizes index $q$, it hence
suffices to find the vertex with maximum index in $\mathcal{P}$.   We do 
this via brute-force search in $O(|\mathcal{P}|)$ time; we will see that
this is $O(1)$ amortized time.

We add $z_q$ to the dominating set $D$, and claim that with this \emph{all}
vertices in $z_1,\dots,z_q$ are dominated.   (Here it will be crucial that we
have an interval graph.)    For if $1\leq p\leq q$, then either $p<i$
(then we made $z_p$ dominated in an earlier iteration), or 
$t(z_q)\leq t(z_p) \leq t(z_i) \leq b(z_q)$
(since $q\geq p\geq i$ and $z_q\in \mathcal{P}$)
and so the intervals of $z_q$ and $z_p$ intersect.

Therefore Operation~A (performed only when $z_i$ was undominated) will be followed
by $q-i$ iterations where Operation~A will \emph{not} be performed.   Also,
$|N_{\geq}(z_i)|\leq q-i+1$ by choice of $q$.
So these ``cheap'' iterations pay for the cost
of performing Operation~A, which hence becomes $O(1)$ time amortized.

\medskip\noindent{\bf Operation B:}   We have to ensure that all vertices
in $N_{<}(z_q)\cup N_{\geq}(z_q)$ are marked ``dominated''.   This is easy
for $N_{<}(z_q)$, a subset of $\{z_1,\dots,z_{q-1}\}$, because we 
simply mark \emph{all} of $z_i,\dots,z_{q-1}$ as dominated (correctness
was argued above).
The time for this is $O(q{-}i)$, which as above is $O(1)$ time amortized.

It remains to find all vertices in $N_{\geq}(z_q)$ that were undominated before
iteration $i$.   Observe that \emph{any} vertex $z_p$ with $p>q$ is undominated at 
this time: We have $z_p\not\in N(z_i)$ (by choice of $q$), and since 
$t(z_p)<t(z_i)$ by $p>q>i$ we therefore have $t(z_p)\leq b(z_p)<t(z_i)$.   On the other
hand, any vertex $z_d$ that is in $D$ before iteration $i$ satisfies $d<i$
(since $z_i$ is undominated), so $t(z_i)<t(z_d)\leq b(z_d)$ and $z_p\not\in N[z_d]$. 
So all we have to compute is $N_{\geq}(z_q)$.    But we saw
the method for this with Operation~A:   This is the same as the contents of
$\mathcal{P}$ at the beginning of iteration $q$.   So rather than implementing
Operation~B, we instead wait until iteration $q$ (all iterations inbetween
are trivial anyway since those vertices are dominated by $z_q$), and in that
iteration, mark the corresponding neighbours as ``dominated''.
Since they were undominated before, this is $O(1)$ amortized time.

\medskip
We summarize the algorithm in the pseudocode in Algorithm~\ref{alg:greedyInterval}.

\begin{theorem} 
Given an $n$-vertex interval graph $G$ in form of an $O(n)$-sized RDV representation $T$ where the host tree is a path, 
a minimum dominating set of $G$ can be found in $O(n)$ time.
\end{theorem}

\begin{algorithm}[ht]
\caption{Dominating set in interval graphs}
\label{alg:greedyInterval}
\Input{Graph $G$ with RDV representation where the host-tree $T$ is a path}
Compute the coordinates of nodes of $T$. \\
Bucket-sort the vertices by decreasing $y(t(\cdot))$ to obtain vertex order $z_1,\dots,z_n$. \\
Bucket-sort the vertices by decreasing of $y(b(\cdot))$ to obtain a list $L$.\\
Initialize dominating set $D = \emptyset$ and an empty list $\mathcal{P}$\\ 
Mark all vertices as ``undominated'' \\

\For{$i = 1,2,\cdots, n$}{
	\While{the first vertex $z_q$ in $L$ satisfies $y(b(z_q))\geq y(t(z_i))$}{
		delete $z_q$ from $L$ and add it to $\mathcal{P}$
	}

	\If{$z_i$ is ``undominated''}{

	find vertex $z_q$ with maximum index in $\mathcal{P}$ \tcp*{Operation A}
       	 $D \leftarrow D \cup \{ z_q \}$ \;
       	 \lFor(\tcp*[f]{Operation B, $N_{<}$}){$j=i$ to $q{-}1$}{mark $z_j$ as ``dominated''} 
	}
	\If(\tcp*[f]{$z_i$ was $z_q$ during last Operation A}){$z_i$ belongs $D$}{
		mark all vertices in $\mathcal{P}$ as ``dominated'' \tcp*{Operation B, $N_{\geq}$ }
	}

    delete $z_i$ from $\mathcal{P}$
}
\KwRet{$D$}
\end{algorithm}

\section{Further Thoughts}
\label{sec:further}

In this paper, we gave a sublinear time algorithm to find the minimum dominating set in an RDV graph.   The main ingredient is to implement the needed operations (in a natural greedy-algorithm) efficiently via geometric data structures involving priority search trees and ray shooting.

Our techniques 
can be applied 
for sublinear time algorithms for other problems.   As an example, consider the independent set problem, i.e., we want to find a maximum set $I$ of vertices such that no two of them are adjacent.   There is a natural greedy algorithm that finds the correct answer in a chordal graph~\cite{Gol80}: Fix a so-called perfect elimination order $z_1,\dots,z_n$ (for an RDV graph, a bottom-up elimination order is a perfect elimination order).   Initialize $I$ as an empty set and mark all vertices as ``not covered''.   For $i=1,\dots,n$, if $z_i$ has not yet been covered, then add $z_i$ to $I$ and mark all vertices in $N[z_i]$ as ``covered''.   
With Operation~B, we can perform in an RDV graph the operation of marking all neighbours as covered in $O(k\log n)$ time, where $k$ is the number of vertices that were not previously covered.    Since this happens to every vertex only once, the run-time becomes $O(n\log n)$ in an RDV graph.
 
We end with an open question. Can the runtime of Theorem~\ref{domsetruntime} be improved to \( O(n \log\log n) \) or even \( O(n) \)? One possible approach is to improve the runtime of priority search trees and orthogonal ray shooting data structures, at least in the special situation when all coordinates are integers in $O(n)$. 
But it is also conceivable that one could generalize the algorithm for interval graphs, i.e., exploit the structure of RDV graphs further to argue that simpler data structures suffice for the operations.


\bibliographystyle{plainurl}
\bibliography{refs} 

\begin{appendix}
\section{Missing Proof}
\NiceRDV*
\begin{proof}
We assume that each node $\nu$ of $T$ knows all vertices for which it is the top node, and all vertices for which it is the bottom node.   (We can compute this in $O(|T|+n)$ time.)   Now perform the following at each node $\nu$ of $T$.   First, if there are $d$ vertices $z_1,\dots,z_d$ for which $\nu$ is the bottom node, then add $d$ new leaves as children of $\nu$, and make these the new bottom nodes for $z_1,\dots,z_d$.   (With this, no top node is a leaf anymore, since the corresponding bottom node became a strict descendant.)    Second, if there are $k\geq 2$ vertices $w_1,\dots,w_k$ for which $\nu$ is the top node, then let $\rho$ be the parent of $\nu$ (temporarily insert a new node as parent $\rho$ if $\nu$ was the root).   Subdivide the link from $\nu$ to $\rho$ with $k-1$ new nodes, and make these the new top nodes of $w_2,\dots,w_k$.   One easily verifies that the newly defined paths intersect in a node if and only if the old paths did.   We added at most two new nodes for each vertex of $G$, and so $|T'|\leq |T|+O(n)$.
\end{proof}
\end{appendix}

\end{document}